\documentclass[french,english,nocjs]{cjs-rcs-article}

\usepackage{natbib}
\usepackage{multirow}
\usepackage{booktabs}
\usepackage{numprint}

\newcommand{\bX}{{\symbf X}}
\newcommand{\bx}{{\symbf x}}

\newcommand{\bbeta}{\symbf{\beta}}
\newcommand{\btheta}{\symbf{\theta}}
\newcommand{\bxi}{ \boldsymbol{\xi} }

\newcommand{\bSigma}{{\symbf \Sigma}}

\newtheorem{example}{Example}
\newtheorem{assumption}{Assumption}
\newtheorem{condition}{Condition}

  \title{A goodness-of-fit test for the logistic propensity score model under nonignorable missing data}

  \author[orcid=0009-0006-1249-0029, email=cmanli176@163.com]
         {Manli \surname{Cheng}}
  \affil{Department of Statistics and Actuarial Science,University of Waterloo, Waterloo, Ontario, Canada}
  \author[orcid=0000-0002-1229-0477, email=yangjc@uwaterloo.ca, corresponding]
         {Yangjianchen \surname{Xu}}
  \affil{Department of Statistics and Actuarial Science,University of Waterloo, Waterloo, Ontario, Canada}
  \author[orcid=0000-0002-9455-5466, email=qinglong.tian@uwaterloo.ca]
         {Qinglong \surname{Tian}}
  \affil{Department of Statistics and Actuarial Science,University of Waterloo, Waterloo, Ontario, Canada}
 \author[orcid=0000-0003-2165-9157, email=pengfei.li@uwaterloo.ca]
         {Pengfei \surname{Li}}
  \affil{Department of Statistics and Actuarial Science,University of Waterloo, Waterloo, Ontario, Canada}

  \begin{englishabstract}
    Logistic regression is widely used to model the propensity score in the analysis of nonignorable missing data. However, goodness-of-fit testing for this propensity score model has received limited attention in the literature. In this paper, we propose a new goodness-of-fit testing procedure for the logistic propensity score model under nonignorable missing data. The proposed test is based on an unweighted sum-of-squared residuals constructed from the marginal missingness mechanism and accommodates the partial observability of the outcome. We establish the asymptotic distribution of the test statistic under both the null hypothesis and general alternatives, and develop a bootstrap procedure with theoretical guarantees to approximate its null distribution. We show that the resulting bootstrap test attains asymptotically correct size and is consistent, with power converging to one under model misspecification. Simulation studies and a real data application demonstrate that the proposed method performs well in finite samples. 
  \end{englishabstract}

  \begin{frenchabstract}
    Logistic regression is widely used to model the propensity score in the analysis of nonignorable missing data. However, goodness-of-fit testing for this propensity score model has received limited attention in the literature. In this paper, we propose a new goodness-of-fit testing procedure for the logistic propensity score model under nonignorable missing data. The proposed test is based on an unweighted sum-of-squared residuals constructed from the marginal missingness mechanism and accommodates the partial observability of the outcome. We establish the asymptotic distribution of the test statistic under both the null hypothesis and general alternatives, and develop a bootstrap procedure with theoretical guarantees to approximate its null distribution. We show that the resulting bootstrap test attains asymptotically correct size and is consistent, with power converging to one under model misspecification. Simulation studies and a real data application demonstrate that the proposed method performs well in finite samples. 
  \end{frenchabstract}



  \begin{sharing}
All computations in this article were carried out using R statistical software. The data and R code required to reproduce the numerical results, including the simulation studies and application examples, are available at the following GitHub repository: \url{https://github.com/ManliCheng/gof-test-ps-mnar}.
  \end{sharing}

  \begin{funding}
Dr. Xu’s work is supported in part by the Natural Sciences and Engineering Research Council of Canada (NSERC, RGPIN-2025-04055).
Dr. Tian’s work is supported by NSERC (RGPIN-2023-03479).
Dr. Li’s work is supported in part by NSERC (RGPIN-2026-04406) and the Faculty of Mathematics Research Chair at the University of Waterloo.
  \end{funding}

\begin{document}

\maketitle                      


\section{Introduction}
\label{chap_intro}

Let $Y$ denote an outcome subject to missingness, $\bX$ a vector of fully observed covariates, and $R$ a missingness indicator taking value $1$ if $Y$ is observed and $0$ otherwise. In nonignorable missing data settings, $R$ may depend on the unobserved value of $Y$ \citep{little2019statistical}. A widely used specification assumes that the propensity score
\(
\pi_t(\bx,y)= \Pr(R=1 \mid Y=y, \bX=\bx)
\)
follows a logistic regression model \citep{ai2020simple,morikawa2021semiparametric,liu2022Full}, often of the form
\begin{equation}
\label{prop.nullmodel}
\pi_0(\bx,y;\alpha,\bbeta,\gamma)
= \frac{1}
{1+\exp\{\alpha + \bbeta^\top r(\bx) + \gamma y\}},
\end{equation}
where $r(\bx)$ is a known function of the covariates, $\alpha$ and $\gamma$ are unknown scalar parameters, and $\bbeta$ is an unknown $m$-dimensional parameter vector. This logistic propensity score model plays a central role in likelihood-based and semiparametric methods for nonignorable missing data, yet its adequacy is rarely assessed in practice. In this paper, we propose a goodness-of-fit test for model \eqref{prop.nullmodel} under nonignorable missing data.

The identification and estimation of parameters in logistic propensity score models have been extensively studied in the literature. For example, \citet{wang2014instrumental} and \citet{miao2016identifiability} investigated identifiability conditions and showed that the model parameters are identifiable when there exists an instrumental variable, that is, a variable that affects the distribution of $Y$ but not the missingness mechanism. A variety of methods have been proposed to estimate the unknown parameters, including generalized method of moments approaches \citep{wang2014instrumental,ai2020simple}, optimal estimating equations \citep{morikawa2021semiparametric}, and likelihood-based methods \citep{morikawa2017semiparametric, liu2022Full,li2023instability,cheng2025novel}. These developments have led to substantial progress in estimation and efficiency theory for nonignorable missing data. Despite these advances, inference procedures based on model \eqref{prop.nullmodel} can be highly sensitive to model misspecification. An incorrect propensity score model may result in biased estimation and misleading scientific conclusions, underscoring the need for formal goodness-of-fit diagnostics.

In the literature, goodness-of-fit testing has been extensively studied for logistic regression models with complete data; see \citet{hosmer1997comparison} for a review. However, these classical procedures are not applicable in the present setting, since the outcome $Y$ is only partially observed and no information on $Y$ is available when $R=0$. 
To our knowledge, only two studies are directly relevant to the goodness-of-fit problem considered here. \citet{chen2020pseudo} considered a special case of \eqref{prop.nullmodel} with $\bbeta=\symbf{0}$. \citet{tang2024model} proposed a goodness-of-fit test for a related but different model in which the linear term $\alpha+\bbeta^\top r(\bx)$ in \eqref{prop.nullmodel} is replaced by a nonparametric function of $\bx$. Their approach has several limitations in the present context. First, it relies on kernel smoothing to estimate the nonparametric component, requiring bandwidth selection and introducing additional tuning complexity. Second, it employs a method-of-moments estimator for $\gamma$, which can be unstable \citep{li2023instability}. Finally, because the null model in \citet{tang2024model} differs from \eqref{prop.nullmodel}, their procedure does not directly address goodness-of-fit for the logistic model considered here.
Given the widespread use of model \eqref{prop.nullmodel} in the analysis of nonignorable missing data, a statistically valid goodness-of-fit procedure tailored specifically to \eqref{prop.nullmodel} is needed.

Our contributions are summarized as follows. First, we show that testing the logistic propensity score model \eqref{prop.nullmodel} under nonignorable missing data can be reduced to testing a parametric specification for the marginal propensity score $\pi_t^*(\bx)=\Pr(R=1\mid \bX=\bx)$. Building on this reduction, we develop a goodness-of-fit test tailored specifically to model \eqref{prop.nullmodel}. The proposed test adapts the unweighted sum-of-squared residual idea from complete-data logistic regression to a setting where the outcome is only partially observed.
Second, we establish the asymptotic distribution of the proposed test statistic under both the null hypothesis and general alternatives. We further construct a bootstrap procedure and show that it achieves asymptotically correct size under the null and consistency under model misspecification.
Third, we demonstrate through simulation studies and a real data application that the proposed procedure performs well in finite samples.

The remainder of the paper is organized as follows. In Section \ref{section2}, we present the problem setup and develop the proposed goodness-of-fit test, along with its asymptotic properties and bootstrap implementation. Section \ref{section.simu} reports simulation studies to evaluate the finite-sample performance of the proposed procedures. Section \ref{sec-real} illustrates the method through a real data application. Section \ref{sec-diss} concludes with a discussion and directions for future research. Technical proofs  are provided in the Appendix.

\section{Main results\label{section2}}
\subsection{Problem setup and parameter estimation}

Suppose that we have $n$ independent and identically distributed observations
$(Y_iR_i, R_i, \bX_i)$ ($i=1,\ldots,n$) from $(YR, R, \bX)$. 
Let
\[
\mathcal{F}_{\text{logit}}
=\left\{\pi_0(\bx,y;\alpha,\bbeta,\gamma)
= \frac{1}{1+\exp\{\alpha + \bbeta^\top r(\bx) + \gamma y\}}:
\alpha \in \mathbb{R},\ \bbeta \in \mathbb{R}^m,\ \gamma \in \mathbb{R}\right\}
\]
denote the class of logistic propensity score models. 
Based on the observed data 
$\{(Y_iR_i,R_i,\bX_i): i=1,\ldots,n\}$, 
our interest is to test
\begin{equation}
\label{test.hypothesis}
H_0: \pi_t(\bx,y) \in \mathcal{F}_{\text{logit}}
\quad \text{versus} \quad
H_A: \pi_t(\bx,y) \notin \mathcal{F}_{\text{logit}}.
\end{equation}

Identification and consistent estimation of the parameters $(\alpha,\bbeta,\gamma)$ under $H_0$ are essential for constructing a valid testing procedure. To this end, we impose the following two assumptions.

\begin{assumption}
There exists an instrumental variable that affects the distribution of $Y$ but does not affect the missingness mechanism $\pi_t(\bx,y)$.
\end{assumption}

\begin{assumption}
The conditional probability density or mass function of $Y$ given the covariates $\bx$ and $R=1$  admits a parametric form $f(y \mid \bx;\bxi)$, where $\bxi$ is an  unknown $p$-dimensional  parameter vector.
\end{assumption}

Assumption~1 ensures that the parameters $(\alpha,\bbeta,\gamma)$ are identifiable under the null hypothesis \citep{wang2014instrumental,miao2016identifiability}. 
Assumption~2 specifies a parametric model for the conditional distribution of $Y$ given $\bx$ among the observed units. Since $(Y,\bX)$ is fully observed when $R=1$, this assumption can be assessed using standard goodness-of-fit procedures; see \citet{fan2001goodness,lin2002model}.
Although Assumption~2 is not required for identification of $(\alpha,\bbeta,\gamma)$ or for consistency of their estimator, estimation may be highly unstable without additional modeling structure, even when the parameters are identifiable \citep{li2023instability}. Taken together, Assumptions~1 and~2 lead to reliable estimation of $(\alpha,\bbeta,\gamma)$ under the null model; see \citet{liu2022Full}.

We next estimate the unknown parameter vector 
$\btheta = (\alpha,\bbeta^{\top},\gamma,\bxi^{\top})^{\top}$
using the maximum likelihood estimation. Let $\pi_t^*(\bx)=\Pr(R=1 \mid \bX=\bx)$. Under Assumption~2 and the null hypothesis of \eqref{test.hypothesis}, it follows from Bayes' formula that $\pi_t^*(\bx)$ admits the following logistic regression representation:
\begin{equation}
\label{new-logistic}
\pi(\bx;\btheta)
=\frac{1}{1+\exp\{\alpha+\bbeta^\top r(\bx)+c(\bx;\gamma,\bxi)\}},
\end{equation}
where 
\(
c(\bx;\gamma,\bxi)=\log\left\{\int \exp(\gamma y)\, f(y \mid \bx;\bxi)\, dy\right\}.
\)

Based on \eqref{new-logistic} and the observed data $\{(Y_iR_i, R_i, \bX_i)\}_{i=1}^n$, the full log-likelihood function is given by
\begin{equation*}
\ell_n(\btheta)
=\sum_{i=1}^n R_i \log f(Y_i \mid \bX_i;\bxi)
+\sum_{i=1}^n R_i \log \pi(\bX_i;\btheta)
+\sum_{i=1}^n (1-R_i)\log\{1-\pi(\bX_i;\btheta)\}.
\end{equation*}
The maximum likelihood estimator of $\btheta$ under the null hypothesis is defined as
\begin{equation}
\label{theta.mle}
\hat{\btheta}
=\arg\max_{\btheta}\, \ell_n(\btheta).
\end{equation}

Note that the log-likelihood $\ell_n(\btheta)$ and the estimator $\hat{\btheta}$ are constructed under the null hypothesis of \eqref{test.hypothesis}. When the data are generated under the alternative, model~\eqref{prop.nullmodel} serves as a working model for $\pi_t(\bx,y)$, and $\ell_n(\btheta)$ and $\hat{\btheta}$ are interpreted as the quasi-log-likelihood and the quasi-maximum likelihood estimator, respectively \citep{white1982maximum}.

Under the regularity conditions given in Section~\ref{conditions} of the Appendix, $\hat{\btheta}$ converges in probability to
\[
\btheta^* = \arg\max_{\btheta} \, \E\{\ell_n(\btheta)\},
\]
which coincides with the true parameter value under the null hypothesis and with the pseudo-true value under the alternative \citep{white1982maximum}.

\subsection{Proposed test}
Under the null hypothesis of \eqref{test.hypothesis} and Assumption~2, it follows from \eqref{new-logistic} that the marginal propensity score $\pi_t^*(\bx)=\Pr(R=1 \mid \bX=\bx)$ belongs to the following class of logistic models:
\[
\mathcal{F}_{\text{logit}}^*
=\left\{\pi(\bx;\btheta)
=\frac{1}{1+\exp\{\alpha+\bbeta^\top r(\bx)+c(\bx;\gamma,\bxi)\}}:
\alpha \in \mathbb{R},\ \bbeta \in \mathbb{R}^m,\ \gamma \in \mathbb{R},\ \bxi \in \mathbb{R}^p\right\}.
\]
Therefore, testing \eqref{test.hypothesis} reduces to testing
\begin{equation}
\label{test.hypothesis2}
H_0: \pi_t^*(\bx) \in \mathcal{F}_{\text{logit}}^*
\quad \text{versus} \quad
H_A: \pi_t^*(\bx) \notin \mathcal{F}_{\text{logit}}^*.
\end{equation}

Let \[
\Delta = \E\left[ \bigl\{R - \pi(\bX;\btheta^*)\bigr\}^2 
- \pi(\bX; \btheta^*)\bigl\{1- \pi(\bX; \btheta^*)\bigr\} \right],
\]
which equals zero under the null and is nonzero under the alternative. Thus, $\Delta$ quantifies the discrepancy between the true propensity score and the logistic propensity score model in \eqref{prop.nullmodel}.
Based on these observations, we propose the following unweighted sum-of-squared residual-based test statistic:
\begin{equation}
\label{Tn}
T_n = \frac{1}{\sqrt{n}} \sum_{i=1}^n
\left[
\bigl\{R_i - \pi(\bX_i;\hat{\btheta})\bigr\}^2
- \pi(\bX_i;\hat{\btheta}) \bigl\{1-\pi(\bX_i;\hat{\btheta})\bigr\}
\right].
\end{equation}

It can be verified that  $T_n/\sqrt{n}$ is consistent to $\Delta$. 
When the null hypothesis of \eqref{test.hypothesis} holds, and hence that of \eqref{test.hypothesis2} also holds, we expect $T_n$ to fluctuate around zero; under the alternative, $|T_n|$ tends to take larger values. Accordingly, we reject the null hypothesis of \eqref{test.hypothesis} for large values of $|T_n|$, with the critical value determined by the limiting distribution of $T_n$.

Note that \cite{hu2025receiver} also proposed $T_n$ for testing the goodness-of-fit of \eqref{prop.nullmodel} when $Y$ is binary. However, they did not study the asymptotic properties of $T_n$ or consider more general response types. This issue is nontrivial, as $T_n$ depends only on the fully observed data $(R,\bX)$, whereas the estimator $\hat{\btheta}$ additionally depends on the observed values of $Y$ when $R=1$. Consequently, existing asymptotic results for unweighted sum-of-squared residual-based tests under complete data are not directly applicable, necessitating new theoretical development.

\subsection{Asymptotic properties}
In this section, we present the asymptotic properties of $T_n$. All proofs are deferred to the Appendix.
The following theorem establishes the asymptotic normality of the test statistic $T_n$, regardless of whether the null hypothesis holds.

\begin{theorem}
\label{null-asynorm}
Under Assumptions~1 and~2, and Conditions 1--7 in the Section~\ref{conditions} of the Appendix, we have
\[
T_n - \sqrt{n}\,\Delta 
\xrightarrow{d} 
\mathcal{N}(0,\sigma^2),
\]
as $n \to \infty$, where
$\xrightarrow{d}$ denotes convergence in distribution,  and $\sigma^2$ is defined in \eqref{final-sigma2} of the Appendix.
\end{theorem}

Based on Theorem 1, a natural approach is to construct a test using a plug-in estimator of $\sigma^2$. With the closed-form expression of $\sigma^2$ given in the Appendix, we may obtain a consistent estimator, denoted by $\hat{\sigma}^2$, for $\sigma^2$. We then reject $H_0$ in \eqref{test.hypothesis} when
\begin{equation}
\label{test.plugin}
|T_n|>\hat{\sigma} z_{1-\alpha/2}
\end{equation}
at the significance level $a$, where $Z_{1-a/2}$ denotes the $(1-a/2)$th quantile of the standard normal distribution.

However, our numerical studies indicate that the plug-in estimator $\hat{\sigma}^2$ tends to overestimate $\sigma^2$, which leads to a conservative test; see \S \ref{section.simu} for details. To address this issue, we propose the following bootstrap procedure, which generates samples under the fitted null model to approximate the null distribution of $T_n$.

\textit{Step 1.} Sample $\bX_1^*,\ldots,\bX_n^*$
from $F_n(\bx)$, the empirical distribution function of $\{\bX_i\}_{i=1}^n$. 

\textit{Step 2.} For $i=1,\ldots,n$, 
generate $R_i^*$ from a $\mathrm{Bernoulli}\{\pi(\bX_i^*;\hat{\btheta})\}$ distribution. 

\textit{Step 3.} For $i=1,\ldots,n$, generate $Y_i^*$ from $f(y \mid \bX_i^*; \hat{\bxi})$ if $R_i^* = 1$. 

\textit{Step 4.} Based on the bootstrap sample $\{(Y_i^*R_i^*, R_i^*, \bX_i^*): i=1,\ldots,n\}$, 
compute the bootstrap estimator $\hat{\btheta}^*$ and the corresponding bootstrap test statistic $T_n^*$
analogously to \eqref{theta.mle} and \eqref{Tn}. 

Based on the bootstrap samples, we approximate the null distribution of $T_n$ by the conditional distribution of $T_n^*$ given the observed data. The following theorem justifies this approach by establishing the limiting conditional distribution of $T_n^*$.

\begin{theorem}
\label{boot-asynorm}
Under Assumptions~1 and~2, and Conditions 1-11 in Section~\ref{conditions} of the Appendix, we have, conditional on the observed data,  
\[
T_n^* 
\xrightarrow{d} 
\mathcal{N}(0,\sigma^2_{\text{Boot}}),
\]
where $\sigma_{\text{Boot}}^2$ is defined in \eqref{sigma_boot} of the Appendix. Furthermore, when $H_0$ of \eqref{test.hypothesis} holds, we have $\sigma_{\text{Boot}}^2 = \sigma^2$.
\end{theorem}

Theorems~\ref{null-asynorm} and~\ref{boot-asynorm} together imply that the null limiting distribution of $T_n$ coincides with the limiting distribution of $T_n^*$ conditional on the observed data. Therefore, we can use the conditional distribution of $T_n^*$ to calibrate the distribution of $T_n$ under $H_0$. Specifically, let $q_{n,1-a}^*$ be the $(1-a)$-quantile of $\bigl|T_n^*\bigr|$ conditional on the observed data. We reject $H_0$ of \eqref{test.hypothesis} when 
\begin{equation}
\label{test2}
|T_n| > q_{n,1-a}^*
\end{equation}
at the significance level $a$. 

The following theorem establishes the asymptotic validity of the proposed bootstrap test.

\begin{theorem}
\label{boot-level-power}
Assume the same conditions as in Theorem~\ref{boot-asynorm}. 

\medskip
\noindent\textnormal{(i) [Asymptotically correct size]} 
If $H_0$ holds, then
\[
\Pr\!\Big(\,\bigl| T_n \bigr| > q_{n,1-a}^*\,\Big)\longrightarrow  a.
\]

\noindent\textnormal{(ii) [Consistency]} 
If $H_A$ holds with $\Delta \neq 0$, then
\[
\Pr\!\Big(\,\bigl| T_n \bigr| > q_{n,1-a}^*\,\Big)\longrightarrow  1.
\]
\end{theorem}

In applications, the bootstrap quantiles $q_{n,1-a}^*$ are obtained via bootstrap simulation. In our simulation study and real data analysis, we use $B=500$ bootstrap samples to approximate the distribution of $T_n^*$ and compute $q_{n,1-a}^*$.

\section{Simulation study\label{section.simu}} 

In this section, we report simulation studies to examine the finite-sample performance of the test in \eqref{test.plugin}, where the asymptotic variance of $T_n$ is estimated by the plug-in method, and of the bootstrap testing procedure in \eqref{test2}, in terms of type I error and power. We refer to the former as the plug-in test and the latter as the bootstrap test.

We consider sample sizes $n = 1000,2000$ and $4000$. 
{Simulated rejection rates under the null models are based on 10,000 replications to ensure accurate estimation of the type I error rate, whereas those under the alternative models are based on 1,000 replications.}
In all scenarios, let $\bX = (X_1, X_2, X_3)^\top$, where
$X_1 \sim \mathcal{N}(0,1)$,
$X_2 \sim \mathrm{Bernoulli}(0.5)$,
$X_3 \sim \mathcal{N}(1,1)$,
and $X_1, X_2, X_3$ are independent. The true propensity score is specified as
\begin{equation*}
\pi_t(\bx,y)
=
\{1 + \exp(\alpha + \beta_1 x_1 + \beta_2 x_2 + \gamma y + e(\bx) + g(\bx)y)\}^{-1},
\end{equation*}
where the terms $e(\bx)$ and $g(\bx)y$ are included to generate alternative models.
We consider three response distributions (binary, normal, and gamma) to examine the performance of the proposed method under different outcome types. For each distribution, five scenarios are studied, with the first corresponding to the null model and the remaining four to alternatives. In all scenarios, the marginal response probability is approximately $\Pr(R = 1) = 0.8$. The variable $X_3$ serves as an instrumental variable and is excluded from the propensity score model.

\begin{example}\label{example-bernoulli}
The response is binary, with
$Y \mid \bX= \bx,\,R=1  \sim \mathrm{Bernoulli}\{p(\bx)\}$,
where
\[
\mbox{logit}\{p(\bx)\}= 1- x_1 - x_2 + 2x_3
\]
with $\mbox{logit}(x)=\log\{x/(1-x)\}$. 
The values of $(\alpha, \beta_1, \beta_2, \gamma)$ and the functions $e(\bx)$ and $g(\bx)$ are given in Table~\ref{tab-para-bn}.
\end{example}

\begin{example}\label{example-norm}
The response is continuous, with
\[
Y \mid \bX=\bx,\, R=1 \sim \mathcal{N}\bigl(\mu(\bx), \sigma^2\bigr),
\]
where
\[
\mu(\bx) = 1 - 1.5 x_1 - 1.5 x_2 + 3 x_3,
\qquad
\sigma^2 = 1.
\]
The corresponding parameter settings are also reported in Table~\ref{tab-para-bn}.
\end{example}

\begin{example}\label{example-gamma}
The response is continuous, with
\[
Y \mid \bX=\bx,\, R=1 \sim \mathrm{Gamma}\bigl(\kappa, \lambda(\bx)\bigr),
\]
where
\[
\lambda(\bx) = \exp(1 - 1.5 x_1 - 1.5 x_2 + 2 x_3).
\]
Here, $\mathrm{Gamma}(a,b)$ denotes the gamma distribution with shape parameter $a$ and scale parameter $b$. The shape parameter is set to $\kappa = 1$ for Scenarios~I--III and 
$\kappa = \exp(1)$ for Scenarios~IV--V. 
This choice of $\kappa$ ensures that the conditional moment generating function 
of $Y$ given $\bX=\bx$ and $R=1$ exists in all scenarios considered in this example.
The values of $(\alpha, \beta_1, \beta_2, \gamma)$, together with the functions $e(\bx)$ and $g(\bx)$, are given in Table~\ref{tab-para-bn}.
\end{example}

\begin{table}[!htt]
  \caption{Parameter values of $(\alpha,\beta_1,\beta_2,\gamma)$, $e(\bx)$ and $g(\bx)$ in Examples \ref{example-bernoulli}-- \ref{example-gamma}.}
  \label{tab-para-bn}
  \centering
  \begin{tabular}{ccccc}
    \toprule
&Scenario & $(\alpha,\beta_1,\beta_2,\gamma)$ & $e(\bx)$ & $g(\bx)$ \\
\midrule
\multirow{6}[0]{*}{Example~\ref{example-bernoulli}} 
 &I   & $(-1.1,-1.5,-1.5,-0.5)$ & $0$ & $0$ \\
 &II  & $(-1.6,-2.0,-2.0,-0.5)$ & $0.5x_1^2$ & $0$ \\
 &III & $(-1.6,-1.5,-2.0,-0.5)$ & $0.5x_1^2 + 0.5x_1^2 x_2$ & $0$ \\
 &IV  & $(-1.0,\phantom{-}1.0,-2.5,-0.5)$ & $0$ & $0.5x_1^2$ \\
 &V   & $(-1.0,-1.0,-2.5,-0.5)$ & $0$ & $0.5x_1^2 + x_1^2 x_2$ \\
 \midrule
\multirow{6}[0]{*}{Example~\ref{example-norm}} 
&  I  & $(-1.0,2.0,-1.0,-0.5)$& $0$ & $0$ \\
&  II & $(-1.0,1.2,-1.0,-0.5)$& $0.5x_1^2$ & $0$ \\
& III & $(-1.0,1.3,-1.5,-0.5)$& $0.5x_1^2 + 0.5x_1^2 x_2$ & $0$ \\
& IV & $(-5.0,1.0,\phantom{-}1.0,-0.5)$ & $0$ & $0.5x_1^2$ \\
& V  & $(-3.5,3.0,-2.0,-0.5)$& $0$ & $0.5x_1^2 + x_1^2 x_2$ \\
\midrule
\multirow{6}[0]{*}{Example~\ref{example-gamma}} 
&
 I    & $(1.0,-1.5,-1.5,-0.5)$ &  $0$ & $0$ \\
  &  II   & $(1.0,-1.5,-2.8,-0.5)$ & $0.5x_1^2$ & $0$ \\
  &  III  & $(1.0,-1.1,-3.5,-0.5)$ & $0.5x_1^2 + 0.5x_1^2 x_2$ & $0$ \\
   & IV   & $(1.0,-1.0,-2.0,-0.5)$   &$0$ &  $0.5 - 0.1\exp(-0.5x_1^2)$\\
  &  V    & $(1.0,-1.0,-2.0,-0.5)$   &$0$ &  $0.5 - 0.1\exp(-x_1^2 + x_2)$\\
    \bottomrule
  \end{tabular}
\end{table}

Table~\ref{table-rr} reports the simulated rejection rates of the plug-in test and the bootstrap test at the $5\%$ significance level for Examples~\ref{example-bernoulli}--\ref{example-gamma}. For Scenario~I, the rejection rates correspond to empirical type~I error, whereas for Scenarios~II--V they correspond to empirical power under different forms of model misspecification.

\begin{table}[!htt]
  \caption{Simulated rejection rates of the bootstrap test and the plug-in test in Examples \ref{example-bernoulli}--\ref{example-gamma}.}
 \label{table-rr}
  \centering
  \begin{tabular}{cccccccc}
	\toprule 
     & & & \multicolumn{5}{c}{Scenarios} \\
     \midrule
		 &  $n$ & Method  &  I & II & III & IV & V \\
         \cmidrule{2-8}
         \multirow{6}[0]{*}{Example~\ref{example-bernoulli}} &
           \multirow{2}[0]{*}{1000} & Bootstrap & 0.054  & 0.610  & 0.999  & 0.739  & 0.778  \\
         & & Plug-in & 0.010  & 0.282  & 0.937  & 0.023  & 0.148  \\
    & \multirow{2}[0]{*}{2000} & Bootstrap & 0.050  & 0.838  & 1.000  & 0.940  & 0.939  \\
         & & Plug-in & 0.024  & 0.622  & 1.000  & 0.219  & 0.436  \\
    &\multirow{2}[0]{*}{4000} & Bootstrap & 0.054  & 0.990  & 1.000  & 0.998  & 0.995  \\
         & & Plug-in & 0.041  & 0.952  & 1.000  & 0.783  & 0.888  \\
         \midrule
     \multirow{6}[0]{*}{Example~\ref{example-norm}} &  
     \multirow{2}[0]{*}{1000} & Bootstrap & 0.053  & 0.462  & 0.874  & 0.564  & 0.975  \\
         &  & Plug-in & 0.013  & 0.128  & 0.547  & 0.000  & 0.053  \\
   & \multirow{2}[0]{*}{2000} & Bootstrap & 0.048  & 0.679  & 0.986  & 0.670  & 1.000  \\
         &  & Plug-in & 0.031  & 0.393  & 0.909  & 0.007  & 0.330  \\
   & \multirow{2}[0]{*}{4000} & Bootstrap & 0.051  & 0.913  & 1.000  & 0.753  & 1.000  \\
        &  & Plug-in & 0.041  & 0.771  & 1.000  & 0.019  & 0.926  \\ 
    \midrule 
     \multirow{6}[0]{*}{Example~\ref{example-gamma}} & 
    \multirow{2}[0]{*}{1000} & Bootstrap & 0.052  & 0.565  & 0.616  & 0.829  & 0.939  \\
        &   & Plug-in & 0.045  & 0.409  & 0.427  & 0.397  & 0.844  \\
  &   \multirow{2}[0]{*}{2000} & Bootstrap & 0.047  & 0.867  & 0.909  & 0.974  & 0.984  \\
    &       & Plug-in & 0.044  & 0.790  & 0.840  & 0.800  & 0.988  \\
   &  \multirow{2}[0]{*}{4000} & Bootstrap & 0.050  & 0.991  & 0.994  & 1.000  & 1.000  \\
       &    & Plug-in & 0.050  & 0.980  & 0.987  & 0.987  & 1.000  \\
    \bottomrule
  \end{tabular}
\end{table}

We first examine performance under the null model. Across all three examples, the bootstrap test yields rejection rates close to the nominal level for all sample sizes, indicating accurate type~I error control in finite samples. In contrast, the plug-in test exhibits different behavior across outcome types. For the binary and normal responses, it is noticeably conservative when $n=1000$ or $2000$, with empirical rejection rates well below the nominal $0.05$ level. Although its type~I error becomes closer to the nominal level when $n=4000$, a noticeable gap remains. For the gamma response, the plug-in test shows improved size performance relative to the other two examples, but it remains less accurate than the bootstrap test, whose rejection rates stay consistently closer to $0.05$.

We next examine power under the alternative models. In all three examples, the rejection rates of both procedures increase with the strength of the alternatives and with sample size, reflecting the expected gain in power as more information becomes available. Moreover, the bootstrap test consistently outperforms the plug-in test. This advantage is particularly pronounced in the binary and normal settings, where the bootstrap test yields substantially higher rejection rates across all four alternative scenarios. In the gamma setting, a similar pattern is observed, although the difference between the two methods is generally less pronounced, especially when the sample size is moderate or large.

Overall, the bootstrap test provides more reliable control of type~I error across a broad range of settings and achieves higher power than the plug-in test under model misspecification. The finite-sample improvement is especially pronounced for binary and normal responses.

\section{Real application\label{section.application}} \label{sec-real}
We analyze data from a study on children’s mental health in Connecticut 
\citep{zahner1997factors, ibrahim2001using}. 
The binary outcome $Y$ indicates teacher-reported psychopathology, where 
$Y=1$ denotes borderline or clinical psychopathology and $Y=0$ denotes normal status. 
The covariates of interest include \emph{father} ($X_1$; $0$ = presence of a father figure, 
$1$ = absence), \emph{health} ($X_2$; $0$ = no health problems, $1$ = fair or poor health, 
chronic condition, or activity limitation), and \emph{parent’s report} ($X_3$; 
$1$ = borderline or clinical psychopathology, $0$ = normal). 
The complete dataset, reported in Table~1 of \citet{ibrahim2001using}, consists of 
2{,}486 children. Teacher reports of psychopathology are missing for 1{,}061 subjects 
(42.7\%), whereas all covariates are fully observed. 

As noted by \citet{ibrahim2001using}, the probability of missingness may depend on the 
child’s unobserved psychopathology status, since teachers may be more inclined to 
report when the status is perceived as clearly normal or clearly abnormal. 
This suggests a nonignorable missingness mechanism. 
\citet{ibrahim2001using} investigated this hypothesis under two model assumptions: 
(i) a logistic outcome regression model for $Y$ given $X_1$ and $X_2$, and 
(ii) a logistic propensity score model for $R$ given $Y$ and $X_2$. 
That is, $X_1$ is treated as an instrumental variable, and $X_3$ is not used in their 
analysis when justifying the nonignorable missingness mechanism. 
Under these two model assumptions, they showed that the coefficient of $Y$ in the 
logistic propensity score model is statistically significant at the 5\% level, 
providing empirical evidence against the missing-at-random assumption.

Next, we apply the proposed test to assess whether the assumed propensity score 
model is supported by the data. Since we observe $(Y,X_1,X_2)$ completely when 
$R=1$, we posit a logistic regression model for the conditional distribution of 
$Y$ given $(X_1,X_2)$ and $R=1$:
\begin{equation}
\label{data.model1}
\mathrm{logit}\{\Pr(Y=1 \mid R=1, X_1, X_2)\} 
= \xi_0 + \xi_1 X_1 + \xi_2 X_2 .
\end{equation}
The estimation results are summarized in Table~\ref{tab:reg-test-summary}. 
A goodness-of-fit test based on the unweighted sum-of-squared residuals for 
logistic regression with complete data \citep{hosmer1997comparison}, which can be 
implemented using the \texttt{resid} function in the \texttt{R} package 
\texttt{rms}, yields a $p$-value of $0.6986$, providing no evidence against the 
adequacy of the assumed logistic model~\eqref{data.model1}. 

\begin{table}[!h]
  \caption{Regression coefficient estimates for models~\eqref{data.model1} and~\eqref{data.model2} 
in the children’s mental health data.}
  \label{tab:reg-test-summary}
  \centering
  \begin{tabular}{lccccclcccc}
  \toprule
&\multicolumn{4}{c}{Model~\eqref{data.model1}}& & &\multicolumn{4}{c}{Model~\eqref{data.model2}} \\
\textbf{Variable} & \textbf{Coef.} & \textbf{S.E.} &\textbf{ Wald $Z$} & \textbf{$p$-value}  && \textbf{Variable} &\textbf{ Coef. }& \textbf{S.E.} & \textbf{Wald $Z$} & \textbf{$p$-value} \\
\cmidrule{1-5} \cmidrule{7-11} 
 Intercept & $-1.737$ & $0.107$ & $-16.240$ & $<0.0001$ & &Intercept & $-1.025$ & $0.689$ & $-1.488$ & $0.137$ \\
$x_1$     & $0.542$  & $0.161$ & $3.370$   & $0.001$  &&$x_2$  & $-0.304$ & $0.122$ & $-2.487$ & $0.013$ \\
$x_2$     & $0.247$  & $0.138$ & $1.790$   & $0.074$  && $y$ & $2.157$  & $1.092$ & $1.976$  & $0.048$ \\
\bottomrule
\end{tabular}
\end{table}

Under model~\eqref{data.model1}, we then apply the proposed bootstrap test and the 
plug-in test to assess the goodness-of-fit of the following logistic propensity 
score model considered in \citet{ibrahim2001using}:
\begin{equation}
\label{data.model2}
\Pr(R=1 \mid Y, X_2) 
= \{1+\exp(\alpha + \beta X_2 + \gamma Y)\}^{-1}.
\end{equation}
The corresponding $p$-values are $0.604$ for the bootstrap test and $0.556$ for 
the plug-in test. These results provide no evidence against the adequacy of the 
propensity score model~\eqref{data.model2} for the children’s mental health data.
The estimation results for model~\eqref{data.model2} are reported in 
Table~\ref{tab:reg-test-summary}. 
These results are broadly consistent with those reported in 
\citet{ibrahim2001using}, indicating similar inferential conclusions.

\section{Conclusion}
\label{sec-diss}
In this paper, we develop a goodness-of-fit test for the logistic propensity score model \eqref{prop.nullmodel} under nonignorable missing data by reducing the problem to testing a parametric specification for the marginal propensity score. This leads to a simple test statistic based on an unweighted sum of squared residuals involving only fully observed variables. 
We establish its asymptotic normality under both the null hypothesis and general alternatives, and propose a bootstrap procedure to approximate the null distribution, which is shown to achieve asymptotically correct size under the null and consistency under alternatives. Simulation studies demonstrate that the bootstrap test provides accurate type~I error control and improved power over the plug-in approach in finite samples, and a real data application illustrates its practical utility.

There are several directions in which the present work could be further developed. For example, it would be of interest to extend the proposed framework to more flexible semiparametric or nonparametric outcome models \citep{li2023instability,cheng2025novel}. It is also of interest to develop goodness-of-fit procedures for broader classes of propensity score models beyond the logistic form, such as the semiparametric propensity score models considered in \citep{kim2011,shao2016}.

\makebackmatter                 

\appendix   
\section{Preparation and regularity conditions}

\subsection{Preparation \label{prepar}}
The technical development involves three types of distributions, which we recall and define below.

Let 
$
F_t(y,r \mid \bx)
$
and $F(\bx)$
denote the true conditional distribution function of $(RY,R)$ given $\bx$
and the marginal distribution function of $\bX$, respectively. 
In what follows, $\E(\cdot)$, $\Var(\cdot)$, and $\Pr(\cdot)$
denote expectation, variance, and probability with respect to the joint distribution 
$F_t(y,r \mid \bx)\, F(\bx)$.

Let 
$
F(y,r \mid \bx,\btheta)
$
denote the conditional distribution determined by $f(y \mid \bx,\bxi)$ and $\pi(\bx;\btheta)$, and define
$
F_0(y,r \mid \bx)
= F(y,r \mid \bx,\btheta^*).
$
Then $F_0(y,r \mid \bx) F(\bx)$ 
is the joint distribution of $(RY,R,\bX)$ under the null hypothesis in \eqref{test.hypothesis}. 
Under the alternative, it is the Kullback–Leibler projection of 
$F_t(y,r \mid \bx) F(\bx)$ onto the class
\[
\left\{
F(y,r \mid \bx,\btheta) F(\bx):
\btheta \in \mathbb{R}^{m+p+2}
\right\}.
\]
In what follows, $\E_0(\cdot)$, $\Var_0(\cdot)$, and $\Pr_0(\cdot)$
denote expectation, variance, and probability with respect to 
$F_0(y,r \mid \bx) F(\bx)$.

Note that $F(y,r \mid \bx,\hat{\btheta}) F_n(\bx)$
is the joint distribution of the bootstrap sample
\[
\{(R_i^* Y_i^*, \bX_i^*, R_i^*)\}_{i=1}^n,
\]
given the observed data.
In what follows, $\E_*$, $\Var_*$, and $\Pr_*$ 
denote expectation, variance, and probability with respect to 
$F(y,r \mid \bx,\hat{\btheta}) F_n(\bx)$.

Let $L_t(y,\bx,r)$ denote the likelihood 
 contribution of $(RY,R)$
conditional on $\bX=\bx$
under $F_t(y,r \mid \bx)$. 
Further, define
$$
L(y,\bx,r;\btheta)
=\{f(y \mid \bx;\bxi)\pi(\bx;\btheta)\}^r
\{1-\pi(\bx;\btheta)\}^{1-r}
$$
and
\begin{equation*}
\ell(y,\bx,r;\btheta)
= r \log f(y \mid \bx;\bxi)
+ r \log \pi(\bx;\btheta)
+ (1-r)\log\{1-\pi(\bx;\btheta)\},
\end{equation*}
which are the likelihood and log-likelihood contributions of $(RY,R)$
conditional on $\bX=\bx$ under $F(y,r \mid \bx,\btheta)$.

Define
\begin{equation*}
\psi(y,\bx,r;\btheta)
= \frac{\partial \ell(y,\bx,r;\btheta)}{\partial \btheta}.
\end{equation*}
Then
\[
\ell_n(\btheta)
= \sum_{i=1}^n \ell(Y_i,\bX_i,R_i;\btheta),
\]
and the estimator $\hat{\btheta}$ satisfies the first-order condition
\[
\frac{\partial \ell_n(\hat{\btheta})}{\partial\btheta}
=
\sum_{i=1}^n
\psi(Y_i,\bX_i,R_i;\hat{\btheta})
= {\symbf{0}}.
\]

Recall that
\[
\btheta^*
=\arg\max_{\btheta} \, \E\{\ell_n(\btheta)\} 
=
\arg\max_{\btheta} \, \E\{\ell(Y,\bX,R;\btheta)\}.
\]
Then $\btheta^*$ satisfies
\begin{equation*}
\E\bigl\{
\psi(Y,\bX,R;\btheta^*)
\bigr\}
= {\symbf{0}}.
\end{equation*}

Define
\begin{equation*}
H(\bx,r;\btheta)
=
\{r - \pi(\bx;\btheta)\}^2
-
\pi(\bx;\btheta)\{1- \pi(\bx;\btheta)\}.
\end{equation*}
Then the test statistic $T_n$ in \eqref{Tn} can be written as
\begin{equation}
\label{def:TnH}
T_n
=
\frac{1}{\sqrt{n}}
\sum_{i=1}^n
H(\bX_i,R_i;\hat{\btheta}).
\end{equation}

\subsection{Regularity conditions\label{conditions}}

The following regularity conditions are imposed to establish the asymptotic properties of $T_n$.

\begin{condition}
$\btheta^*$ is the unique maximizer of $\E\{\ell_n(\btheta)\}$ under the alternative.
\end{condition}

\begin{condition}
The parameter space ${\symbf{\Theta}}$ of ${\btheta}$ is compact, and 
${\btheta^*}$ lies in the interior of ${\symbf{\Theta}}$.
\end{condition}

\begin{condition}
The function $\ell(y,\bx,r;\btheta)$ has continuous third-order partial derivatives with respect to $\btheta$ for all $(y,\bx,r)$.
\end{condition}

\begin{condition}
(i) $\E\{|\log L_t(Y,\bX,R)|\}<\infty$; 
(ii) for sufficiently small $\rho>0$,
\[
\E\!\left[
\log\{1+L(Y,\bX,R;\btheta,\rho)\}
\right]<\infty
\quad \text{for all } \btheta \in {\symbf{\Theta}},
\]
where
\[
L(y,\bx,r;\btheta,\rho)
=
\sup_{\|\btheta' -\btheta\|<\rho}
L(y,\bx,r;\btheta').
\]
Here $\|\btheta\|$ denotes the $L_2$-norm of $\btheta$.
\end{condition}

\begin{condition}
The matrix
\[
{\symbf{J}}
=
-\E\!\left\{
\frac{\partial^2\ell(Y,\bX,R;\btheta^*)}
{\partial\btheta\,\partial\btheta^\top}
\right\}
\]
has full rank, and
\[
{\symbf{\Sigma}}
=
\Var\!\left\{
Z(Y,\bX,R;\btheta^*)
\right\}
\]
is positive definite, where
\[
Z(Y,\bX,R;\btheta^*)
=
\begin{pmatrix}
H(\bX,R;\btheta^*)\\
\psi(Y,\bX,R;\btheta^*)
\end{pmatrix}.
\]
\end{condition}

\begin{condition}
The elements of
\[
{\symbf{h}}
=
\E\!\left\{
\frac{\partial H(\bX,R;\btheta^*)}{\partial \btheta}
\right\}
\]
exist and are finite.
\end{condition}

\begin{condition}
Let $\theta_1$, $\theta_2$, and $\theta_3$ be any three components of ${\btheta}$. 
There exist $\epsilon>0$ and a function $M(y,\bx,r)$ such that
\[
\sup_{\|\btheta-\btheta^*\|\le \epsilon}
\left|
\frac{\partial^3 \ell(y,\bx,r;\btheta)}
{\partial\theta_1\partial\theta_2\partial\theta_3}
\right|
\le M(y,\bx,r),
\qquad
\sup_{\|\btheta-\btheta^*\|\le \epsilon}
\left|
\frac{\partial^2 H(\bx,r;\btheta)}
{\partial\theta_1\partial\theta_2}
\right|
\le M(y,\bx,r),
\]
and $\E\{M(Y,\bX,R)\}<\infty$.
\end{condition}

The asympototic properties of $T_n^*$ further relies on the following reguarltiy conditions.

\begin{condition}
The bootstrap estimator $\hat{\btheta}^*$ is consistent to $\btheta^*$. \end{condition}

\begin{condition}
The matrices 
$$
{\symbf{J}}_0=-\E_0\left\{
\frac{\partial^2\ell(Y,\bX,R;\btheta^*)}{\partial\btheta\partial\btheta^\top}
\right\}
\qquad
\mbox{and}
\qquad
{\symbf{\Sigma}}_0=\Var_0\left\{
Z(Y,\bX,R;\btheta^*)
\right\} $$ 
have full rank. 
Further, 
the elements of
\[
{\symbf{h}}_0
=
\E_0
\left\{ 
\frac{\partial H(\bX,R;\btheta^*)}{\partial \btheta}
\right\}
\]
exist and are finite.
\end{condition}

\begin{condition}
Let $\theta_1$ and $\theta_2$ be any two components of $\btheta$. 
With $\epsilon$ and $M(y,\bx,r)$ defined in Condition~7, there exists a positive $\epsilon_0$  such that 
\[
\sup_{\|\btheta-\btheta^*\|\le \epsilon}
\left|
\frac{\partial \ell(y,\bx,r;\btheta)}
{\partial\theta_1}
\right|^{2+\epsilon_0}
\le M(y,\bx,r),
\qquad
\sup_{\|\btheta-\btheta^*\|\le \epsilon}
\left|
\frac{\partial^2 \ell(y,\bx,r;\btheta)}
{\partial\theta_1\partial\theta_2}
\right|
\le M(y,\bx,r),
\]
and 
\[
\sup_{\|\btheta-\btheta^*\|\le \epsilon}
\left|
\frac{\partial H(\bx,r;\btheta)}
{\partial\theta_1}
\right|
\le M(y,\bx,r).
\]
Movever, 
\[
\sup_{\|\btheta-\btheta^*\|\le \epsilon} M^*(\bx;\btheta)
\le M^{**}(\bx),
\]
where
\[
M^*(\bx;\btheta)
=
\iint_{y,r} M^2(y,\bx,r)\, dF(y,r \mid \bx,\btheta),
\]
and $\E\{M^{**}(\bX)\}<\infty$.
\end{condition}

\begin{condition}
Let $m(y,\bx,r;\btheta)$ be any element of
\[
\left\{
\frac{\partial^2 \ell(y,\bx,r;\btheta)}
{\partial\btheta\,\partial \btheta^\top}, 
\;
\frac{\partial H(\bx,r;\btheta)}
{\partial\btheta}, 
\;
Z(y,\bx,r;\btheta)Z^{\top}(y,\bx,r;\btheta)
\right\}.
\]
Define
\[
m^*(\bx;\btheta)
=
\iint_{y,r} m(y,\bx,r;\btheta)
\, dF(y,r \mid \bx,\btheta).
\]
Then $m^*(\bx;\btheta)$ is continuously differentiable in $\btheta$, and
\[
\sup_{\|\btheta-\btheta^*\|\le \epsilon}
\left|
\frac{\partial m^*(\bx;\btheta)}
{\partial \theta_1}
\right|
\le M^{**}(\bx),
\]
for any component $\theta_1$ of $\btheta$, where
$\epsilon$ is defined in Condition~7, 
$M^{**}(\bx)$ is defined in Condition~10, and
$\E\{M^{**}(\bX)\}<\infty$.
\end{condition}

We comment that Conditions 1–7 are standard regularity conditions for establishing asymptotic properties of parametric models. In particular, they guarantee that $\hat{\btheta}$ is strongly consistent for $\btheta^*$ and asymptotically normal \citep{white1982maximum}. Moreover, these conditions ensure that the test statistic $T_n$ admits a valid linear approximation, which is essential for deriving its limiting distribution.

The consistency assumption in Condition 8 can be established by adapting the Argmax theorem, that is, Theorem 5.7 of \cite{van2000asymptotic}. More discussions can be found in \cite{cheng2010}. Conditions 9–11 are standard conditions to ensure that the conditional weak law of large numbers and conditional central limit theorem can be applied to some bootstrap quantities involved in the proof of Theorem 2 in the main paper.

\section{Proof of Theorem 1 in the main paper\label{proof.thm1}}

Recall that $\hat{\btheta}$ solves the following estimating equation
\[
\sum_{i=1}^n \psi(Y_i,\bX_i,R_i;\hat{\btheta}) = {\symbf{0}}.
\]
Following \citet{white1982maximum}, under Conditions~1--7, a Taylor expansion of the estimating equation around $\btheta^*$ yields
\begin{align}
\label{theta_asy1}
\sqrt{n}(\hat{\btheta} - \btheta^*)
&=
-\left\{
\frac{1}{n}\sum_{i=1}^n
\frac{\partial^2 \ell(Y_i,\bX_i,R_i; \btheta^*)}
{\partial\btheta\,\partial\btheta^\top}
\right\}^{-1}
\frac{1}{\sqrt{n}}
\sum_{i=1}^{n}
\psi(Y_i,\bX_i,R_i;\btheta^*)
+ o_p(1).
\end{align}

By the weak law of large numbers,
\[
-\frac{1}{n}\sum_{i=1}^n
\frac{\partial^2 \ell(Y_i,\bX_i,R_i ; \btheta^*)}
{\partial\btheta\,\partial\btheta^\top}
\;\xrightarrow\;
{\symbf{J}}
\]
in probability, 
where
\[
{\symbf{J}}
=
\E
\left[
-
\frac{\partial^2 \ell(Y,\bX,R; \btheta^*)}
{\partial\btheta\,\partial\btheta^\top}
\right]
\]
is nonsingular
by Condition~5. 

Therefore, by the continuous mapping theorem and Slutsky’s theorem, \eqref{theta_asy1} implies
\begin{align}
\label{theta_asy2}
\sqrt{n}(\hat{\btheta} - \btheta^*)
&=
{\symbf{J}}^{-1}
\frac{1}{\sqrt{n}}
\sum_{i=1}^{n}
\psi(Y_i,\bX_i,R_i;\btheta^*)
+ o_p(1).
\end{align}

Recall that 
\begin{equation*}
H(\bx,r;\btheta)
=
\{r - \pi(\bx;\btheta)\}^2
-
\pi(\bx; \btheta)\{1- \pi(\bx; \btheta)\},
\end{equation*}
and 
$$
T_n
=
\frac{1}{\sqrt{n}}
\sum_{i=1}^n
H(\bX_i,R_i;\hat{\btheta}).
$$
Then
\[
\Delta
=
\E\{ H(\bX,R;\btheta^*) \}
\qquad
\mbox{and}
\qquad
T_n - \sqrt{n}\Delta
=
\frac{1}{\sqrt{n}}
\sum_{i=1}^n
\left\{
H(\bX_i,R_i;\hat{\btheta})
-
\Delta
\right\}.
\]

Applying a first-order Taylor expansion of $H(\bX_i,R_i;\hat{\btheta})$ around $\btheta^*$ gives
\begin{align*}
T_n - \sqrt{n}\Delta
&=
\frac{1}{\sqrt{n}}
\sum_{i=1}^n
\left\{
H(\bX_i,R_i;\btheta^*)
-
\Delta
\right\}
\\
&\quad
+
\frac{1}{n}
\sum_{i=1}^n
\left\{
\frac{\partial H(\bX_i,R_i;\btheta^*)}{\partial \btheta}
\right\}^{\!\top}
\sqrt{n}(\hat{\btheta}-\btheta^*)
+
o_p(1).
\end{align*}

Recall that 
\[
{\symbf{h}}
=
\E
\left\{
\frac{\partial H(\bX,R;\btheta^*)}{\partial \btheta}
\right\}.
\]
By the weak law of large numbers, 
\[
\frac{1}{n}
\sum_{i=1}^n
\frac{\partial H(\bX_i,R_i;\btheta^*)}{\partial \btheta}
\;\xrightarrow{}\;
{\symbf{h}}
\]
in probability. 
Combining this with \eqref{theta_asy2} and applying Slutsky’s theorem yields
\begin{align*}
T_n-\sqrt{n}\Delta
&=
\frac{1}{\sqrt{n}}\sum_{i=1}^n
\Big[
H(\bX_i,R_i;\btheta^*)
-\Delta
+
{\symbf{h}}^\top{\symbf{J}}^{-1}
\psi(Y_i,\bX_i,R_i;\btheta^*)
\Big]
+ o_p(1).
\end{align*}
Under Conditions~1--7, the summand has mean zero and finite second moment. Hence, by the central limit theorem and Slutsky’s theorem,
\[
T_n-\sqrt{n}\Delta
\;\xrightarrow{d}\;
\mathcal{N}(0,\sigma^2),
\]
where
\begin{align}
\label{final-sigma2}
\sigma^2
&=
(1,{\symbf{h}}^\top {\symbf{J}}^{-1})\bSigma
(1,{\symbf{h}}^\top {\symbf{J}}^{-1})^\top
.
\end{align}
This completes the proof.

\section{Proof of Theorem 2 in the main paper\label{proof.thm2}}

\subsection{Two useful lemmas}
We first present two useful lemmas. The first concerns a conditional weak law of large numbers, and the second establishes a conditional central limit theorem for certain bootstrap quantities. 

In what follows, we use $(Y^*,R^*,\bX^*)$ to denote random elements drawn from the distribution $F(y,r\mid\bx,\hat{\btheta})F_n(\bx)$.

\begin{lemma}
\label{lemma1}
Assume Conditions~1--11. Then the following results hold:

\begin{itemize}

\item[(i)] 
\[
n^{-1}\sum_{i=1}^n M(Y_i^*,\bX_i^*,R_i^*)
-
\E_*\{M(Y^*,\bX^*,R^*)\}
=
o_p(1),
\]
and
\[
n^{-1}\sum_{i=1}^n M(Y_i^*,\bX_i^*,R_i^*)
=
O_p(1).
\]

\item[(ii)] 
\[
n^{-1}\sum_{i=1}^n
\frac{\partial^2\ell(Y_i^*,\bX_i^*,R_i^*;\hat{\btheta})}
{\partial\btheta\,\partial\btheta^\top}
=
-{\symbf{J}}_0
+
o_p(1).
\]

\item[(iii)]
\[
n^{-1}\sum_{i=1}^n
\frac{\partial H(Y_i^*,\bX_i^*,R_i^*;\hat{\btheta})}
{\partial\btheta}
=
{\symbf{h}}_0
+
o_p(1).
\]

\end{itemize}
\end{lemma}

\begin{proof}
We consider part~(i). Following \citet{cheng2010}, for the first statement it suffices to show that for any $\epsilon_1>0$ and $\epsilon_2>0$,
\[
\Pr\Bigg[
{\Pr}_*\Bigg(
\left|
\frac{1}{n}\sum_{i=1}^n M(Y_i^*,\bX_i^*,R_i^*)
-
\E_*\{M(Y^*,\bX^*,R^*)\}
\right|
\ge \epsilon_1
\Bigg)
\ge \epsilon_2
\Bigg]
\to 0.
\]

By Chebyshev's inequality,
\begin{align}
{\Pr}_*\Bigg(
\left|
\frac{1}{n}\sum_{i=1}^n M(Y_i^*,\bX_i^*,R_i^*)
-
\E_*\{M(Y^*,\bX^*,R^*)\}
\right|
\ge \epsilon_1
\Bigg)
&\le
\frac{\Var_*\!\left\{ M(Y^*,\bX^*,R^*) \right\}}
{n\epsilon_1^2}
\notag\\
&\le
\frac{\E_*\!\left\{ M^2(Y^*,\bX^*,R^*) \right\}}
{n\epsilon_1^2}.
\label{lemma1.eq0}
\end{align}
Under Condition~10 and using \eqref{theta_asy2}, we have
\begin{align}
\E_*\{M^2(Y^*,\bX^*,R^*)\}
&=
\int_{\bx}
\iint_{y,r}
M^2(y,\bx,r)
\, dF(y,r\mid\bx,\hat{\btheta})\,
dF_n(\bx)
\nonumber\\
&=
\frac{1}{n}\sum_{i=1}^n
M^*(\bX_i;\hat{\btheta})
\nonumber\\
&\le
\frac{1}{n}\sum_{i=1}^n
M^{**}(\bX_i),
\label{lemma1.eq1}
\end{align}
where $M^*(\bx;\btheta)
=
\iint_{y,r} M^2(y,\bx,r)\, dF(y,r\mid\bx,\btheta)$
and $M^{**}(\bx)$ is the dominating function in Condition~10.

Combining \eqref{lemma1.eq1} with \eqref{lemma1.eq0}, we obtain
\[
{\Pr}_*\Bigg(
\left|
\frac{1}{n}\sum_{i=1}^n M(Y_i^*,\bX_i^*,R_i^*)
-
\E_*\{M(Y^*,\bX^*,R^*)\}
\right|
\ge \epsilon_1
\Bigg)
\le
\frac{n^{-1}\sum_{i=1}^n M^{**}(\bX_i)}
{n\epsilon_1^2}.
\]

Hence,
\begin{align*}
&\Pr\Bigg(
\Pr_*\Bigg(
\left|
\frac{1}{n}\sum_{i=1}^n M(Y_i^*,\bX_i^*,R_i^*)
-
\E_*\{M(Y^*,\bX^*,R^*)\}
\right|
\ge \epsilon_1
\Bigg)
\ge \epsilon_2
\Bigg)
\\
&\le
\Pr\left(
\frac{n^{-1}\sum_{i=1}^n M^{**}(\bX_i)}
{n\epsilon_1^2}
>
\epsilon_2
\right)
\;\to\; 0,
\end{align*}
where the last step follows from Condition~10 and the weak law of large numbers.
This completes the proof of the first statement in part~(i), that is,
\begin{equation}
\label{lemma1.part1.first}
n^{-1}\sum_{i=1}^n
M(Y_i^*,\bX_i^*,R_i^*)
-
\E_*\{M(Y^*,\bX^*,R^*)\}
=
o_p(1).
\end{equation}

For the second statement in part~(i), by \eqref{lemma1.eq1} and Condition~10, we have
\[
0
\le
\E_*\{M(Y^*,\bX^*,R^*)\}
\le
\left[\E_*\{M^2(Y^*,\bX^*,R^*)\}\right]^{1/2}
\le
\left[
\frac{1}{n}\sum_{i=1}^n M^{**}(\bX_i)
\right]^{1/2}
= O_p(1),
\]
where the second inequality follows from the Cauchy--Schwarz inequality and the last step follows from Condition~10 and the weak law of large numbers.
Combining this with \eqref{lemma1.part1.first}, we obtain
\[
n^{-1}\sum_{i=1}^n
M(Y_i^*,\bX_i^*,R_i^*)
=
O_p(1).
\]
This completes the proof of part~(i). 

We now consider part~(ii). 
Arguing as in the proof of \eqref{lemma1.part1.first}, we obtain
\begin{align*}
n^{-1}\sum_{i=1}^n
\frac{\partial^2\ell(Y_i^*,\bX_i^*,R_i^*;\hat{\btheta})}
{\partial\btheta\,\partial\btheta^\top}
=
\E_*\left\{
\frac{\partial^2\ell(Y^*,\bX^*,R^*;\hat{\btheta})}
{\partial\btheta\,\partial\btheta^\top}
\right\}
+
o_p(1).
\end{align*}

Note that
\begin{align*}
\E_*\left\{
\frac{\partial^2\ell(Y^*,\bX^*,R^*;\hat{\btheta})}
{\partial\btheta\,\partial\btheta^\top}
\right\}
&=
\int_{\bx}
\iint_{y,r}
\frac{\partial^2\ell(y,\bx,r;\hat{\btheta})}
{\partial\btheta\,\partial\btheta^\top}
\, dF(y,r\mid\bx,\hat{\btheta})\,
dF_n(\bx).
\end{align*}

By \eqref{theta_asy2} and Condition~11 applied to the elements of
\[
\frac{\partial^2\ell(y,\bx,r;\btheta)}
{\partial\btheta\,\partial\btheta^\top},
\]
a first-order expansion around $\btheta^*$ yields
\begin{align*}
&\int_{\bx}
\iint_{y,r}
\frac{\partial^2\ell(y,\bx,r;\hat{\btheta})}
{\partial\btheta\,\partial\btheta^\top}
\, dF(y,r\mid\bx,\hat{\btheta})\,
dF_n(\bx)
\nonumber\\
&=
\int_{\bx}
\iint_{y,r}
\frac{\partial^2\ell(y,\bx,r;\btheta^*)}
{\partial\btheta\,\partial\btheta^\top}
\, dF(y,r\mid\bx,\btheta^*)\,
dF_n(\bx)
+
O_p(n^{-1/2}).
\end{align*}

Hence,
\begin{align*}
\E_*\left\{
\frac{\partial^2\ell(Y^*,\bX^*,R^*;\hat{\btheta})}
{\partial\btheta\,\partial\btheta^\top}
\right\}
&=
\int_{\bx}
\iint_{y,r}
\frac{\partial^2\ell(y,\bx,r;\btheta^*)}
{\partial\btheta\,\partial\btheta^\top}
\, dF(y,r\mid\bx,\btheta^*)\,
dF_n(\bx)
+
o_p(1).
\end{align*}

By the weak law of large numbers,
\begin{align}
\E_*\left\{
\frac{\partial^2\ell(Y^*,\bX^*,R^*;\hat{\btheta})}
{\partial\btheta\,\partial\btheta^\top}
\right\}
&=
\int_{\bx}
\iint_{y,r}
\frac{\partial^2\ell(y,\bx,r;\btheta^*)}
{\partial\btheta\,\partial\btheta^\top}
\, dF(y,r\mid\bx,\btheta^*)\,
dF(\bx)
+
o_p(1)
\nonumber\\
&=
-{\symbf{J}}_0
+
o_p(1).
\label{lemma1.part2.J0}
\end{align}

This completes the proof of part~(ii).

The proof of part~(iii) proceeds analogously to that of part~(ii). 
This completes the proof of the lemma.
\end{proof}

\begin{lemma}
\label{lemma2}
Assume Conditions~1--11. Conditional on the observed data, 
\[
\frac{1}{\sqrt{n}}
\sum_{i=1}^n
Z(Y_i^*,\bX_i^*,R_i^*;\hat{\btheta})
\;\xrightarrow{d}\;
\mathcal{N}\bigl(
{\symbf{0}}, \bSigma_{0}
\bigr).
\]
\end{lemma}
\begin{proof}
Since $F(y,r\mid \bx,\hat{\btheta})F_n(\bx)$
is the joint distribution of the bootstrap sample
\[
\{(R_i^* Y_i^*, \bX_i^*, R_i^*)\}_{i=1}^n,
\]
it follows that
\[
\E_*\!\left\{
Z(Y^*,\bX^*,R^*;\hat{\btheta})
\right\}
=
{\symbf{0}}.
\]

Let
\[
\bSigma_n^*
=
\Var_*\!\left\{
Z(Y^*,\bX^*,R^*;\hat{\btheta})
\right\}
=
\E_*\!\left[
Z(Y^*,\bX^*,R^*;\hat{\btheta})\,
Z^{\top}(Y^*,\bX^*,R^*;\hat{\btheta})
\right].
\]
Similarly to the proof of \eqref{lemma1.part2.J0}, we obtain
\begin{equation}
\label{lemma2.bsigma}
\bSigma_n^*
=
\bSigma_0
+
o_p(1).
\end{equation}

Conditions~10 and~11 verify the Lyapunov condition (p.~69 of \citealp{shao2003}), thereby ensuring that the conditional central limit theorem applies.
That is, conditional on the observed data,
\[
(\bSigma_n^*)^{-1/2}
\frac{1}{\sqrt{n}}
\sum_{i=1}^n
Z(Y_i^*,\bX_i^*,R_i^*;\hat{\btheta})
\;\xrightarrow{d}\;
\mathcal{N}\bigl(
{\symbf{0}}, {\symbf{I}}
\bigr),
\]
where ${\symbf{I}}$ denotes the identity matrix.

Together with \eqref{lemma2.bsigma} and the conditional Slutsky theorem \citep{cheng2015}, it follows that
\[
\frac{1}{\sqrt{n}}
\sum_{i=1}^n
Z(Y_i^*,\bX_i^*,R_i^*;\hat{\btheta})
\;\xrightarrow{d}\;
\mathcal{N}\bigl(
{\symbf{0}}, \bSigma_0
\bigr).
\]

This completes the proof of the lemma.
\end{proof}
\subsection{Proof of Theorem 2
\label{proof.thm2.main}}

With the help of Lemmas~\ref{lemma1} and~\ref{lemma2}, we prove Theorem~2 in three steps.

\medskip
\noindent
\textbf{Step 1.} Derive a linear representation of the bootstrap estimator $\hat{\btheta}^*$.
\medskip

Note that $\hat{\btheta}^*$ solves
\[
\sum_{i=1}^n 
\psi(Y_i^*,\bX_i^*,R_i^*;\hat{\btheta}^*)
=
{\symbf{0}}.
\]
By Condition~8 and \eqref{theta_asy2}, we have
\[
\hat{\btheta}^*-\hat{\btheta}=o_p(1).
\]
A first-order Taylor expansion of $\psi(Y_i^*,\bX_i^*,R_i^*;\btheta)$ around $\hat{\btheta}$ yields
\begin{equation}
\label{step1.eq0}
{\symbf{0}}
=
\sum_{i=1}^n 
\psi(Y_i^*,\bX_i^*,R_i^*;\hat{\btheta})
+
\sum_{i=1}^n
\frac{\partial^2\ell(Y_i^*,\bX_i^*,R_i^*;\hat{\btheta})}
{\partial\btheta\,\partial\btheta^\top}
(\hat{\btheta}^*-\hat{\btheta})
+
\epsilon_n^*,
\end{equation}
where $\epsilon_n^*$ is the remainder term satisfying
\[
\|\epsilon_n^*\|
\le
(p+m+2)^2
\sum_{i=1}^n
M(Y_i^*,\bX_i^*,R_i^*)\,
\|\hat{\btheta}^*-\hat{\btheta}\|^2.
\]

By part~(i) of Lemma~\ref{lemma1},
\[
\sum_{i=1}^n
M(Y_i^*,\bX_i^*,R_i^*)
=
O_p(n),
\]
and hence,
\begin{equation}
\label{step1.eq1}
\epsilon_n^*
=
o_p(n)\,\|\hat{\btheta}^*-\hat{\btheta}\|.
\end{equation}
Using part~(ii) of Lemma~\ref{lemma1}, we also have
\begin{equation}
\label{step1.eq2}
\sum_{i=1}^n
\frac{\partial^2\ell(Y_i^*,\bX_i^*,R_i^*;\hat{\btheta})}
{\partial\btheta\,\partial\btheta^\top}
=
-n{\symbf{J}}_0
+
o_p(n).
\end{equation}
Lemma~\ref{lemma2} further implies that
\begin{equation}
\label{step1.eq3}
\frac{1}{\sqrt{n}}
\sum_{i=1}^{n}
\psi(Y_i^*,\bX_i^*,R_i^*;\hat{\btheta})
=
O_p(1).
\end{equation}

Plugging \eqref{step1.eq1}--\eqref{step1.eq3} into \eqref{step1.eq0} and solving for $\hat{\btheta}^*-\hat{\btheta}$ yield the linear approximation of $\hat{\btheta}^*$:
\begin{equation}
\label{thetaB_asy2}
\sqrt{n}(\hat{\btheta}^* - \hat{\btheta}) 
=
{\symbf{J}}_0^{-1}
\frac{1}{\sqrt{n}}
\sum_{i=1}^{n}
\psi(Y_i^*,\bX_i^*,R_i^*;\hat{\btheta})
+
o_p(1).
\end{equation}

\medskip\noindent
\textbf{Step 2.} Derive the stochastic expansion of the bootstrap statistic $T_n^*$.
\medskip

Similarly, we perform a first-order Taylor expansion of $T_n^*$ around $\hat{\btheta}$:
\begin{equation}
\label{testB-expan1}
T_n^*
=
\frac{1}{\sqrt{n}}\sum_{i=1}^{n} H(\bX_i^*,R_i^*;\hat{\btheta})
+
\frac{1}{\sqrt{n}}\sum_{i=1}^{n}
\left\{
\frac{\partial H(\bX_i^*,R_i^*;\hat{\btheta})}
{\partial\btheta}
\right\}^{\!\top}
(\hat{\btheta}^*-\hat{\btheta})
+
\epsilon_{n}^{**},
\end{equation}
where the remainder term $\epsilon_{n}^{**}$ satisfies
\[
|\epsilon_{n}^{**}|
\le
\frac{(p+m+2)^2}{\sqrt{n}}
\sum_{i=1}^n
M(Y_i^*,\bX_i^*,R_i^*)
\,
\|\hat{\btheta}^*-\hat{\btheta}\|^2.
\]

By \eqref{thetaB_asy2}, part~(i) of Lemma~\ref{lemma1}, and Lemma~\ref{lemma2}, we further obtain
\begin{equation}
\label{step2.eq1}
\epsilon_{n}^{**}
=
O_p(n^{-1/2})
=
o_p(1).
\end{equation}
By part~(iii) of Lemma~\ref{lemma1}, we also have
\begin{equation}
\label{step2.eq2}
\frac{1}{n}
\sum_{i=1}^{n}
\frac{\partial H(\bX_i^*,R_i^*;\hat{\btheta})}
{\partial\btheta}
=
{\symbf{h}}_0
+
o_p(1).
\end{equation}

Plugging \eqref{thetaB_asy2} and \eqref{step2.eq1}--\eqref{step2.eq2} into \eqref{testB-expan1} yields
\begin{equation}
\label{testB-expan2}
T_n^*
=
\frac{1}{\sqrt{n}}
\sum_{i=1}^{n}
\left\{
H(\bX_i^*,R_i^*;\hat{\btheta})
+
{\symbf{h}}_0^\top
{\symbf{J}}_0^{-1}
\psi(Y_i^*,\bX_i^*,R_i^*;\hat{\btheta})
\right\}
+
o_p(1).
\end{equation}
\medskip

\noindent
\textbf{Step 3.} Establish the limiting distribution of $T_n^*$ conditional on the observed data.
\medskip

Note that $T_n^*$ in \eqref{testB-expan2} can be rewritten as
\begin{equation}
\label{testB-expan3}
T_n^*
=
(1,{\symbf{h}}_0^\top {\symbf{J}}_0^{-1})
\frac{1}{\sqrt{n}}
\sum_{i=1}^{n}
Z(Y_i^*,\bX_i^*,R_i^*;\hat{\btheta})
+
o_p(1).
\end{equation}

By Lemma~\ref{lemma2} and the conditional Slutsky theorem \citep{cheng2015}, it follows that, conditional on the observed data,
\[
T_n^*
\xrightarrow{d}
\mathcal{N}(0,\sigma^2_{\mathrm{Boot}}),
\]
where
\begin{equation}
\label{sigma_boot}
\sigma^2_{\mathrm{Boot}}
=
(1,{\symbf{h}}_0^\top {\symbf{J}}_0^{-1})
\bSigma_0
(1,{\symbf{h}}_0^\top {\symbf{J}}_0^{-1})^\top.
\end{equation}

Under $H_0$, we have
\[
{\symbf{h}} = {\symbf{h}}_0, 
\qquad
{\symbf{J}} = {\symbf{J}}_0, 
\qquad
{\bSigma} = {\bSigma}_0.
\]
Hence, under $H_0$,
\[
\sigma^2 = \sigma^2_{\mathrm{Boot}}.
\]
This completes the proof of Theorem~2.

\section{Proof of Theorem 3 in the main paper\label{proof.thm3}}

Let $q_{1-a}(\sigma^2)$ and $q_{1-a}(\sigma^2_{\mathrm{Boot}})$
denote the $(1-a)$-quantiles of $N(0,\sigma^2)$ and
$N(0,\sigma^2_{\mathrm{Boot}})$, respectively.

By Theorem~2 and Problem~23.1 of \citet{van2000asymptotic}, we have
\[
\sup_x
\left|
{\Pr}_*\!\left(T_n^* \le x\right)
-
\Pr\!\left\{ N(0,\sigma^2_{\mathrm{Boot}}) \le x \right\}
\right|
=
o_p(1).
\]
Together with Lemma~21.2 of \citet{van2000asymptotic}, it follows that 
\begin{equation}
\label{quan.eq1}
q_{n,1-a}^*
=
q_{1-a}(\sigma^2_{\mathrm{Boot}})
+
o_p(1).
\end{equation}

For (i), under $H_0$, we have $\sigma^2=\sigma^2_{\mathrm{Boot}}$.
Then \eqref{quan.eq1} implies
\begin{equation}
\label{quan.eq2}
q_{n,1-a}^*
=
q_{1-a}(\sigma^2)
+
o_p(1).
\end{equation}
Hence,
\begin{align*}
\Pr\!\left(|T_n|>q_{n,1-a}^*\right)
&=
\Pr\!\left(
|T_n|>
q_{1-a}(\sigma^2)
+
o_p(1)
\right).
\end{align*}

Recall that under $H_0$, $\Delta=0$. By Theorem~1,
\[
T_n \xrightarrow{d} N(0,\sigma^2).
\]
Hence, by the continuous mapping theorem,
\[
|T_n| \xrightarrow{d} |N(0,\sigma^2)|.
\]
Together with \eqref{quan.eq2} and Slutsky’s theorem, we obtain
\[
\Pr\!\left(|T_n|>q_{n,1-a}^*\right)
\longrightarrow
\Pr\!\left(|N(0,\sigma^2)|>q_{1-a}(\sigma^2)\right)
=
a,
\]
as $n\to\infty$.

For (ii), assume that $H_A$ holds and $\Delta\neq 0$.
Without loss of generality, assume $\Delta>0$. Then
\begin{align}
\label{partii.eq1}
\Pr\!\left(|T_n|>q_{n,1-a}^*\right)
&\ge
\Pr\!\left(T_n>q_{n,1-a}^*\right) =
\Pr\!\left(
T_n-\sqrt{n}\Delta
>
q_{1-a}(\sigma^2_{\mathrm{Boot}})
-\sqrt{n}\Delta
+
o_p(1)
\right).
\end{align}

By Theorem~1, 
\[
T_n-\sqrt{n}\Delta
\xrightarrow{d}
N(0,\sigma^2).
\]
Hence, by Slutsky’s theorem and Lemma~2.11 of \citet{van2000asymptotic},
\[
\delta_n
=
\sup_x
\left|
\Pr\!\left(T_n-\sqrt{n}\Delta -o_p(1)\le x\right)
-
\Pr\!\left(N(0,\sigma^2)\le x\right)
\right|
=
o(1).
\]
Therefore, from \eqref{partii.eq1},
\begin{align*}
\Pr\!\left(|T_n|>q_{n,1-a}^*\right)
&\ge
\Pr\!\left(
N(0,\sigma^2)
>
q_{1-a}(\sigma^2_{\mathrm{Boot}})
-
\sqrt{n}\Delta
\right)
-
\delta_n.
\end{align*}

Since $\Delta>0$, we have 
\[
q_{1-a}(\sigma^2_{\mathrm{Boot}})
-
\sqrt{n}\Delta
\;\longrightarrow\;
-\infty,
\]
and therefore
\[
\Pr\!\left(
N(0,\sigma^2)
>
q_{1-a}(\sigma^2_{\mathrm{Boot}})
-
\sqrt{n}\Delta
\right)
\longrightarrow 1.
\]
Consequently,
\[
\Pr\!\left(|T_n|>q_{n,1-a}^*\right)
\longrightarrow 1
\qquad\text{as } n\to\infty.
\]
This completes the proof. 


\bibliography{gof_ref}
\end{document}